\documentclass[11pt,letterpaper]{article}
\usepackage[margin=1in]{geometry}
\usepackage[T1]{fontenc}
\usepackage{lmodern,microtype}
\usepackage{amsmath,amssymb,amsthm,mathtools,booktabs}
\usepackage[dvipsnames]{xcolor}
\usepackage[colorlinks=true,linkcolor=MidnightBlue,citecolor=MidnightBlue,urlcolor=MidnightBlue]{hyperref}
\usepackage[capitalize,nameinlink,noabbrev]{cleveref}
\let\originalthebibliography\thebibliography
\renewcommand{\thebibliography}[1]{%
  \originalthebibliography{#1}%
  \setlength{\itemsep}{2pt}%
  \setlength{\parsep}{0pt}%
}
\newtheorem{theorem}{Theorem}[section]
\newtheorem{lemma}[theorem]{Lemma}

\theoremstyle{definition}
\newtheorem{definition}[theorem]{Definition}
\newtheorem{remark}[theorem]{Remark}
\newcommand{\F}{\mathbb F}

\newcommand{\TV}{\operatorname{TV}}
\newcommand{\Shr}{\operatorname{Shr}}
\newcommand{\Rec}{\operatorname{Rec}}
\newcommand{\Ext}{\operatorname{Ext}}
\newcommand{\err}{\operatorname{err}}

\newcommand{\K}{\mathcal K}
\newcommand{\M}{\mathcal M}
\newcommand{\Sour}{\mathcal S}

\newcommand{\ind}{\mathbf 1}
\newcommand{\Cb}{\bar C}
\hypersetup{pdftitle={Secret Sharing from Non-Extractable Randomness},pdfauthor={}}
\title{Secret Sharing Despite Non-Extractable Randomness}
\author{Divesh Aggarwal\\
		National University of Singapore\\
		\texttt{dcsdiva@nus.edu.sg}}
\date{\today}
\begin{document}
\maketitle
\begin{abstract}
\noindent
Bosley and Dodis (TCC 2007) proved that a class of imperfect randomness
sources that suffices for information-theoretic encryption of messages
longer than the logarithm of the key length must admit deterministic
extraction of nearly uniform bits, and asked whether the same is true of
two-out-of-two secret sharing. Aggarwal, Chung, Obremski, and Ribeiro
(TCC 2022) gave a positive answer for schemes that are in addition
leakage-resilient or non-malleable; the original question remained open.

We answer it in the negative. For every $b\ge2$ we construct a class of
flat sources on $\{0,1\}^{b+2}$, each of min-entropy at least $b$, together
with a perfectly correct and perfectly private two-out-of-two scheme that
shares a $b$-bit secret into two $(b+1)$-bit shares using any source in the
class, and we show that every deterministic map from $\{0,1\}^{b+2}$ to
$\{0,1\}$ is at statistical distance at least $1/6$ from uniform on some
source in the class. This implies that sharing a long enough secret does not require extractable randomness.
The proof is a short parity argument. The idea of the construction extends to other thresholds: for
two-out-of-$n$ and $r$-out-of-$r$ sharing we give perfect schemes with
$(b+1)$-bit shares, private for classes of flat sources from which no single
bit can be extracted with error below $1/(4n-2)$ and $1/(2(2^r-1))$,
respectively.
\end{abstract}

\section{Introduction}
Uniform randomness is the most fundamental resource in information-theoretic
cryptography. A uniform one-time pad gives perfect encryption, and splitting a secret
into a uniform one-time pad and the padded secret gives perfect two-out-of-two secret
sharing~\cite{Sha49,Sha79,Bla79}. In contrast, physical sources of randomness are typically not
uniform, and their exact distribution is usually unknown; the most one can
assume is that the distribution belongs to some well-defined \emph{class} of
sources, for instance the class of all sources with a given min-entropy. The
cleanest way to use such a source is to extract uniform randomness from it: if a single
deterministic function maps every source in the class to a nearly uniform
string, every cryptographic task requiring uniform randomness becomes possible. A basic question, going
back to~\cite{MP90,DS02,DOPS04}, is whether extraction is necessary:
\begin{quote}
\emph{If a class of sources suffices for a given cryptographic task, must it
be extractable?}
\end{quote}
The answer depends on the task. Entropy alone does not suffice:
deterministic extraction fails for sources with almost full
min-entropy~\cite{SV84,CG88}, and such sources do not support
information-theoretic privacy either~\cite{MP90,DOPS04}. For carefully
chosen classes, on the other hand, privacy is possible without extraction
when the message is a single bit: Dodis and Spencer separated one-bit
encryption from extraction~\cite{DS02}, and Dodis, Pietrzak, and Przydatek
separated one-bit secret sharing from one-bit encryption~\cite{DPP06}.

\paragraph{Encryption implies extraction.}
Bosley and Dodis~\cite{BD07} showed that these separations only work if the message length is small. If a class of distributions on $d$-bit keys permits
perfect encryption of $b$-bit messages, then for every $\varepsilon>0$ with
\[
 b>\log d+2\log(1/\varepsilon)
\]
the class admits a deterministic extractor with
$\lfloor b-2\log(1/\varepsilon)\rfloor$ output bits and error
$\varepsilon$~\cite[Theorem~1]{BD07}. They complemented this result with
non-extractable classes that support encryption of
$b\le\log d-\log\log d-2$ bits. Hence, as soon as the message is sufficiently longer than
the logarithm of the key length, \emph{encryption requires extractable
randomness}.

Encryption yields two-out-of-two secret sharing (give the key to one party
and the ciphertext to the other), but not conversely~\cite{DPP06}, so secret
sharing is the weaker and more basic primitive. Bosley and Dodis left open
whether it obeys the same principle~\cite[Section~5]{BD07}:
\begin{quote}
\emph{Does two-out-of-two secret sharing of long secrets require extractable
randomness?}
\end{quote}
Aggarwal, Chung, Obremski, and Ribeiro~\cite{ACOR22} revisited the question
fifteen years later. They proved that the answer is \emph{yes} when the
scheme additionally resists one bit of local leakage from each share, or is
non-malleable (unless $d$ is exponentially large); see also~\cite{Chu24}.
They posed the question for general thresholds and explicitly left ordinary
secret sharing open.

\subsection{Our results}
We resolve the question in the negative, in a strong form: the scheme is
perfect, the secret is only two bits shorter than the source string, every
source is flat with high min-entropy, and not even one bit can be extracted
with error below a fixed constant.
 
A \emph{source class} is a fixed family of distributions on the dealer's
random input, which we call the \emph{key}. A single sharing scheme must be
private whichever distribution in the class the key is drawn from, and a
single deterministic extractor must work for every distribution in the
class. Extraction error is measured by statistical distance from uniform.
Formal definitions are in \cref{sec:model}.
 
\begin{theorem}[Main result]\label{thm:main}
For every integer $b\ge2$ there are a source class $\Sour_b$ on
$\{0,1\}^{b+2}$ and a perfectly correct two-out-of-two secret sharing scheme
for $b$-bit messages that is perfectly private for $\Sour_b$, such that:
\begin{enumerate}
\item each share has $b+1$ bits;
\item every distribution in $\Sour_b$ is uniform on a set of size between
$2^b$ and $2^{b+1}$, and hence has min-entropy at least $b$;
\item for every $\ell\ge1$ and every function
$\Ext:\{0,1\}^{b+2}\to\{0,1\}^\ell$,
\[
 \max_{P\in\Sour_b}\TV(\Ext(K),U_\ell)\ge\frac16,\qquad K\sim P.
\]
\end{enumerate}
Sharing and reconstruction use a constant number of operations in
$\F_{2^b}$.
\end{theorem}
 
The extractor in item~3 is arbitrary: it may be computationally unbounded
and may depend on the whole key. The key has $d=b+2$ bits, so the secret is
far longer than $\log d$, which is the regime in which encryption forces
extraction~\cite{BD07}. The constant $1/6$ cannot be improved for this
class: a suitable function of the key has error at most $1/6$ on every
source in $\Sour_b$ (\cref{thm:twon} with $n=2$).
 
\paragraph{Other thresholds.}
Aggarwal, Chung, Obremski, and Ribeiro~\cite{ACOR22} also asked the
question for $t$-out-of-$n$ sharing. Their random-less reductions turn any $t$-out-of-$n$ scheme with
$t\ge2$ into a two-out-of-two scheme for the same source class, so a
separation for any threshold gives one for two-out-of-two, and
two-out-of-two is the easiest case in which to separate sharing from
extraction. The reductions do not go in the other direction (see the end of
\cref{sec:rr}), so \cref{thm:main} by itself says nothing about other
thresholds. We extend it to two-out-of-$n$ sharing for every $n$, and to
$r$-out-of-$r$ sharing for every $r$.
 
\begin{theorem}[see \cref{thm:twon,thm:rr}]\label{thm:higher-intro}
\leavevmode
\begin{enumerate}
\item For every $n\ge2$ and every $b$ with $2^b>n$, there are a source class
on a key set of size $(2n-1)2^b$ and a perfectly correct
two-out-of-$n$ scheme for $b$-bit messages, with $(b+1)$-bit shares, that is
perfectly private for this class, such that every function $\Ext$ to
$\{0,1\}$ has statistical error at least $1/(4n-2)$ on some distribution in
the class.
\item For every $r\ge2$ and every $b\ge r$, there are a source class on
a key set of size $(2^r-1)2^{(r-1)b}$ and a perfectly correct $r$-out-of-$r$ scheme  for $b$-bit
messages, with $(b+1)$-bit shares, that is perfectly private for this class,
such that every function $\Ext$ to $\{0,1\}$ has statistical error at least
$1/(2(2^r-1))$ on some distribution in the class.
\end{enumerate}
In both cases every distribution in the class is flat, with min-entropy at
least $\log|\K|-\log3$, where $\K$ is the key set, and the bound on the error
is tight for the class.
\end{theorem}

\subsection{Technical overview}

For the purpose of the technical overview, we focus on two-out-of-two sharing. 
\paragraph{Why the argument for encryption does not transfer.}
The proof of Bosley and Dodis considers, for each key $k$, the distribution
of the ciphertext of a uniform message under $k$. Decryptability makes each
of these $2^d$ distributions flat with min-entropy $b$, and so a random function
$f$ extracts from all these distributions simultaneously when $b>\log d$ and hence from $\mathrm{Enc}(U_b, K)$ where $K$ is a valid key source; privacy lets
one replace the uniform $b$-bit message $U_b$ by the fixed message $0^b$, and
$k\mapsto f(\mathrm{Enc}(0^b,k))$ is the extractor. 

Let us see which part of this idea carries over to two-out-of-two secret
sharing, and where it breaks down. Let the two shares of a $b$-bit message
$m$ be $\Shr_1(m,K)$ and $\Shr_2(m,K)$. Given $\Shr_2(U_b,K)$, the share
$\Shr_1(U_b,K)$ determines $U_b$, which is independent of $\Shr_2(U_b,K)$ by
privacy; hence $\Shr_1(U_b,K)$ has min-entropy at least $b$. By privacy,
$\Shr_1(0^b,K)$ has the same distribution as $\Shr_1(U_b,K)$, so it too has
min-entropy at least $b$. To extract from it with a single function, as
Bosley and Dodis do, one would write its distribution, for every source in
the class, as a convex combination of a fixed and ``small'' family of
high-min-entropy sources. For encryption this family is
$\{\mathrm{Enc}(U_b,k)\}_k$; for secret sharing we see no analogous family in
general. Unfortunately, any attempt to try to do this has to restrict the share size or the key size to be nearly $b$.

For secret sharing, the
object that is private (one share) and the object that determines the
message (both shares) are different. A single share, as a function of the
message, need not be injective, and may even be constant.

Our counter-example exploits exactly this difference as we will explain below. 

\paragraph{The three-key gadget.}
Let $V=\F_{2^b}$ (which has the same additive structure as $\F_2^b$) and let
$\alpha\notin\{0,1\}$ be any fixed element of $\F_{2^b}$. For every index $i\in V$ there are three
keys. The \emph{central} key $C_i$ shares $m$ as $(i+m,\alpha \cdot i+m)$, where addition and multiplication are in $\F_{2^b}$; each share is a
one-time-padded copy of $m$, and together they determine $m$. The \emph{left}
key $L_i$ gives the first party the same share $i+m$ and gives the
second party a dummy $\bar i$, a tagged copy of the index $i$. The tagged values $\{\bar i: i \in \F_{2^b}\}$ are distinct from the untagged
values in $\F_{2^b}$. The \emph{right}
key $R_i$ is symmetric, and gives the first party the share $\bar i$, and the second party the share $\alpha \cdot i + m$. The sources in our class are the uniform distributions on sets of keys that
contain, for each $i$, either the single key $C_i$ or the pair
$\{L_i,R_i\}$; there are $2^{2^b}$ of them.

Let us first see why this scheme is a perfect 2-out-of-2 secret sharing scheme. Clearly, given either (i) $i+m$ and $\alpha \cdot i + m$, or (ii) $i+m$ and $\bar i$, or (iii) $\bar i$ and $\alpha \cdot i + m$, it is straightforward to recover $m$. 

Now we argue that neither share reveals any information about the message.
Consider a source $K$ in the class, uniform on a set $Z$ of keys. Let
$T \subseteq \F_{2^b}$ be the set of indices $i$ with $C_i\in Z$, so that
$\F_{2^b} \setminus T$ is the set of indices $i$ with $L_i,R_i\in Z$.

Then, for any fixed message $m$, the first share is uniformly distributed over the set
\[
\{i+m: i \in T\} \; \cup \; \{i + m: i \in \F_{2^b} \setminus T\} \; \cup \; \{\bar i: i \in \F_{2^b} \setminus T\} \;,
\]
where the three sets correspond to center keys, left keys, and right keys. This implies that the first share is uniform over $\F_{2^b} \cup \{\bar i: i \in \F_{2^b} \setminus T\}$, and hence is independent of the message $m$. 

The argument for why the second share is independent of the message $m$ is identical.

\paragraph{Parity defeats every extractor.}
We now explain why there is no deterministic extractor that extracts even a single bit from all such randomness sources. 
For convenience, we will view a candidate one-bit extractor as a labelling $h$ of the keys by $\{-1, +1\}$. 

Notice that a valid randomness source $K$ is random over some set of keys $Z$, and this set $Z$ contains, for each index $i$, either $C_i$, or both $L_i$ and $R_i$. Thus, for $h$ to be an extractor, the output of $h$ must be (almost) unbiased regardless of whether $Z$ contains $C_i$ or both $L_i$ and $R_i$ for any index $i$. In other words, we want $$\frac{\sum_{k \in Z} h(k)}{|Z|}$$ to be small.  

We will argue that for any given $h$, there exists a choice of $Z$ for which the above mentioned bias is a constant.

At index $i$ the two allowed choices contribute $h(C_i)$ or $h(L_i)+h(R_i)$
to the signed sum. Their difference $h(C_i) - h(L_i) - h(R_i)$ is a sum of three $\pm 1$ values, and is hence odd, i.e. one of $\{+1,+3,-1,-3\}$. 

Let $Q$ contain $C_i$ for every index $i$ with
$h(C_i)-h(L_i)-h(R_i)>0$, and both $L_i$ and $R_i$ for every other index, and
let $Q^c$ be the set of the remaining keys, which makes the other choice at
every index. Both $Q$ and $Q^c$ are valid choices of $Z$, and
\[
\sum_{k \in Q} h(k) - \sum_{k \in Q^c} h(k)
=\sum_{i}\bigl|h(C_i)-h(L_i)-h(R_i)\bigr| \ge 2^b,
\qquad |Q| + |Q^c| = 3\cdot2^b.
\]
So one of the two sources has normalized signed sum at least $1/3$ in
absolute value, that is, statistical error at least $1/6$.

\subsection{Related work}
The separation of Dodis, Pietrzak, and Przydatek~\cite{DPP06} shares a single
bit; its entropy amplification enlarges the randomness space while retaining
that message space. The small-message counterexamples for encryption
in~\cite{BD07} yield sharing counterexamples as well, so the long-secret
regime is essential to our claim. Randomness requirements of other
primitives have been studied too: Santha--Vazirani and min-entropy sources
fail already for one-bit encryption~\cite{MP90,DOPS04}, computational
assumptions can partially substitute for extraction~\cite{KLR09}, and
in~\cite{ACMPS14}, the authors prove impossibility results for encryption, commitment and zero knowledge under efficient online tampering with the random tape.

Dodis and Yao~\cite[Definition~9, Theorem~4(c), and Remark~2]{DY15} derive a
different, weak extraction consequence from secret sharing. Their extractor
may output a failure symbol, and the lower bound on its success probability
is required under the uniform input distribution rather than under every
allowed source. This does not give the deterministic extraction considered
here. The positive results of~\cite{ACOR22} require leakage resilience or
non-malleability; \cref{rem:leak} shows the one-bit leakage attack on our
scheme that their theorem predicts.

Yao and Li~\cite{YL25} state an extraction implication for ordinary threshold
secret sharing. Their Proposition~2 extends the encryption reduction by
applying privacy to the complete tuple of shares, whereas ordinary sharing
hides only unauthorized subsets and the full tuple determines the message.
\Cref{app:comparison} identifies the step that fails, and shows that the
proposition itself is false for sharing.

\paragraph{Organization.}
\Cref{sec:model} contains definitions. \Cref{sec:construction} proves
\cref{thm:main}. \Cref{sec:twon} considers two-out-of-$n$ sharing and
\cref{sec:rr} extends the result to $r$-out-of-$r$ secret sharing. 

\section{Definitions}\label{sec:model}
All logarithms are base two. For a
nonempty finite set $A$, write $U_A$ for the uniform distribution on $A$, and
abbreviate $U_{\{0,1\}^\ell}$ by $U_\ell$. For distributions $P,P'$ on the
same finite set, the total variational distance (also called statistical distance, or statistical error) is defined as
\[
 \TV(P,P')=\frac12\sum_x|P(x)-P'(x)|=\max_A|P(A)-P'(A)|.
\]
For a random variable $X$ we write $\TV(X,P')$ for the distance between the
distribution of $X$ and $P'$. The min-entropy of $P$ is
$H_\infty(P)=-\log\max_xP(x)$. A \emph{flat} distribution is uniform on a
nonempty subset of its domain.

\begin{definition}[Threshold sharing from a source class]\label{def:ss}
Let $\K$ be a finite key set, $\M$ a message set, and $\Sour$ a class of
distributions on $\K$. A \emph{$t$-out-of-$n$ sharing scheme} $\Pi$ consists
of a deterministic map
\[\Shr=(\Shr_1,\dots,\Shr_n):\M\times\K\to\mathcal C_1\times\dots\times
\mathcal C_n\] together with a reconstruction map for every set of at least
$t$ parties. The dealer takes one sample $K$ from an unknown $P\in\Sour$ and
outputs $\Shr(m,K)$; there is no other randomness. For $J\subseteq[n]$ write
$\Shr_J$ for the shares of the parties in $J$.
\begin{itemize}
\item $\Pi$ is \emph{perfectly correct} if, for every $J$ with $|J|\ge t$,
all $m\in\M$, and all $k\in\K$, the reconstruction map $\Rec_J$ satisfies
\[
 \Rec_J(\Shr_J(m,k))=m .
\]
\item $\Pi$ is \emph{$\delta$-private for $\Sour$} if for every $P\in\Sour$,
every $J$ with $|J|<t$, and all $m,m'\in\M$,
\[
 \TV\bigl(\Shr_J(m,K),\Shr_J(m',K)\bigr)\le\delta,\qquad K\sim P.
\]
It is \emph{perfectly private} if $\delta=0$.
\end{itemize}
For $t=n=2$ we write $F=\Shr_1$ and $G=\Shr_2$.
\end{definition}
The distribution of a share may depend on $P$; it must not depend on the
message.

\begin{definition}[Extraction error]\label{def:ext}
For $\Ext:\K\to[D]$ and a class $\Sour$ of distributions on $\K$, let
\[
 \err(\Ext,\Sour)=\sup_{P\in\Sour}\TV\bigl(\Ext(K),U_{[D]}\bigr),
 \qquad K\sim P.
\]
The class is \emph{$(\ell,\varepsilon)$-extractable} if some
$\Ext:\K\to\{0,1\}^\ell$ has $\err(\Ext,\Sour)\le\varepsilon$.
\end{definition}
This is the formulation of~\cite{BD07,ACOR22}. Composing an $\ell$-bit
extractor with the projection to its first bit does not increase the error,
so a lower bound on one-bit error applies to every output length.

\section{The construction}\label{sec:construction}
Let $b \ge 2$, $V=\F_{2^b}$ and $B=2^b$. Fix an $\alpha\in \F_{2^b} \setminus \{0,1\}$. 

The key set is $\K=\{L,R,C\}\times V$, and we write $L_i,R_i,C_i$ for
$(L,i),(R,i),(C,i)$. Both share alphabets are $\{0,1\}\times V$; we write $x$
for $(0,x)$ and $\bar x$ for $(1,x)$, and call the latter \emph{tagged}. The
bar is a formal tag, not an operation on $V$. The message set is $\M=V$, and
the scheme $\Pi_b$ shares $m\in V$ as follows:
\begin{equation}\label{eq:sharing}
\begin{array}{c|cc}
\text{key }k&F(m,k)&G(m,k)\\\hline
L_i&i+m&\bar i\\
R_i&\bar i&\alpha \cdot i+m\\
C_i&i+m&\alpha \cdot i+m
\end{array}
\end{equation}
For every $T\subseteq V$ define
\begin{equation}\label{eq:sources}
 \Omega_T=\{C_i:i\in T\}\cup\{L_i,R_i:i\notin T\},\qquad
 Z_T=|\Omega_T|=2B-|T|,\qquad P_T=U_{\Omega_T}.
\end{equation}
The source class is $\Sour_b=\{P_T:T\subseteq V\}$. It is fixed once and for
all, independently of the message and of any candidate extractor. Since
$B\le Z_T\le2B$, every $P_T$ is flat with min-entropy at least $b$.

\subsection{Correctness, privacy, and non-extractability}
\begin{lemma}[Correctness]\label{lem:correct}
There is a map $\Rec$ with $\Rec(F(m,k),G(m,k))=m$ for all $m\in V$ and
$k\in\K$.
\end{lemma}
\begin{proof}
The pattern of tags identifies the branch. If only the second share is
tagged, the shares are $(x,\bar i)$ with $x=i+m$, and $\Rec$ outputs $x+i$.
If only the first share is tagged, the shares are $(\bar i,y)$ with
$y=\alpha \cdot i+m$, and $\Rec$ outputs $y+\alpha \cdot i$. If neither share is tagged, the shares
are $(x,y)=(i+m,\alpha \cdot i+m)$, so $y+ \alpha \cdot x =(\alpha + 1) \cdot m$ and $\Rec$ outputs
$(\alpha \cdot x + y)/(\alpha + 1)$.
\end{proof}

\begin{lemma}[Privacy]\label{lem:private}
For every $T\subseteq V$ and every $m\in V$, if $K\sim P_T$ then each of
$F(m,K)$ and $G(m,K)$ is uniform on
\begin{equation}\label{eq:marginal}
 V\;\cup\;\{\bar i:i\notin T\}.
\end{equation}
In particular $\Pi_b$ is perfectly private for $\Sour_b$.
\end{lemma}
\begin{proof}
Fix $T$ and $m$. For each $i$, exactly one of $L_i,C_i$ lies in $\Omega_T$,
and either key has first share $i+m$. As $i\mapsto i+m$ is a permutation of
$V$, each untagged value has probability exactly $1/Z_T$. The tagged first
shares arise from the keys $R_i$ with $i\notin T$; each value $\bar i$ with
$i\notin T$ has probability $1/Z_T$. 

For the second share the argument is
the same: exactly one of $R_i,C_i$ lies in $\Omega_T$, either key has second
share $\alpha \cdot i+m$, and $i\mapsto \alpha \cdot i+m$ is a permutation because $\alpha \neq 0$; the tagged values come from $L_i$, $i\notin T$. The
set~\eqref{eq:marginal}, on which each share is uniform over, does not depend on $m$.
\end{proof}

\begin{lemma}[No bit can be extracted]\label{lem:noext}
For every $\Ext:\K\to\{0,1\}$ there is $T\subseteq V$ such that $\Ext$ has
statistical error at least $1/6$ on $P_T$ or on $P_{T^c}$, where
$T^c=V\setminus T$.
\end{lemma}
\begin{proof}
Let $h(k)=(-1)^{\Ext(k)}$ and, for $T\subseteq V$, let
\[H_T:=\sum_{k\in\Omega_T}h(k) = \sum_{i \in T} h(C_i) + \sum_{i \in T^c} \Big(h(L_i)+ h(R_i)\Big)\;.\] For $K\sim P_T$ the distance of $\Ext(K)$
from $U_1$ equals $|H_T|/(2Z_T)$. For each $i\in V$ the integer
\[
 v_i=h(C_i)-h(L_i)-h(R_i)
\]
is a sum of three odd numbers and is therefore odd; in particular
$|v_i|\ge1$. Let $T=\{i:v_i>0\}$. Then, $T^c = \{i: v_i < 0\}$. In particular, for each $i \in T$, $v_i$ is either $1$ or $3$, and for each $i$ in $T^c$, $v_i$ is either $-1$, or $-3$. Then
\[
 H_T-H_{T^c}=\sum_{i\in T}v_i-\sum_{i\notin T}v_i=\sum_{i\in V}|v_i|\ge B,
 \qquad Z_T+Z_{T^c}=3B.
\]
Consequently
\begin{equation}\label{eq:bias}
 \max\left\{\frac{|H_T|}{Z_T},\frac{|H_{T^c}|}{Z_{T^c}}\right\}
 \ge\frac{|H_T|+|H_{T^c}|}{Z_T+Z_{T^c}}\ge\frac{H_T-H_{T^c}}{3B}\ge\frac13,
\end{equation}
and the source corresponding to the maximum ($\Omega_T$ or $\Omega_{T^c}$) has error at least $1/6$.
\end{proof}

\subsection{Bit strings and running time}\label{sec:encoding}
\begin{proof}[Proof of \cref{thm:main}]
Encode a key by two branch bits followed by the $b$-bit index, and a share
by its tag bit followed by $b$ bits. Three of the four branches are $L,R,C$.
To define the scheme on all of $\{0,1\}^{b+2}$, let the fourth branch consist
of keys $\Cb_i$, tagged copies of the central keys, with
\[
 \Shr(m,\Cb_i)=\bigl(\,\overline{i+m},\,\overline{\alpha \cdot i+m}\,\bigr);
\]
on two tagged shares, $\Rec$ removes the tags and applies the rule of the
central branch. Correctness on $\K$ is \cref{lem:correct}. Every $P_T$
assigns probability zero to the fourth branch, so \cref{lem:private} gives
privacy. The set $T$ in \cref{lem:noext} depends only on the restriction of
the extractor to $\K$, so the lemma applies verbatim to functions on
$\{0,1\}^{b+2}$, and the remark after \cref{def:ext} extends it to every
output length. Items~1 and~2 are immediate.

The running time for sharing and reconstruction algorithms is immediate from the construction. 
\end{proof}

\begin{remark}[Consistency with the leakage-resilient case]\label{rem:leak}
By~\cite{ACOR22}, a class that supports a two-out-of-two scheme resisting
one bit of leakage from each share is extractable, so $\Pi_b$ must be
vulnerable to such leakage. It is. Write $\operatorname{lsb}(z)$ for the
least significant bit of $z\in V$ in a fixed basis over $\F_2$, which is
$\F_2$-linear in $z$. For untagged shares $x$ and $y$ of the two parties let
$f(x)=\operatorname{lsb}(\alpha\cdot x)$ and $g(y)=\operatorname{lsb}(y)$
(on tagged shares $f$ and $g$ may be arbitrary). Under the source $P_V$,
which is uniform on the central keys, the two leaked bits satisfy
\[
 f(i+m)+g(\alpha\cdot i+m)=\operatorname{lsb}\bigl((\alpha+1)\cdot m\bigr)
 \qquad\text{for every } i,
\]
which is a nontrivial linear functional of the secret. Hence the leakage
distributions of two secrets on which this functional differs have disjoint
supports.
\end{remark}

\section{Two-out-of-\texorpdfstring{$n$}{n} sharing}\label{sec:twon}
We now extend the construction of \cref{sec:construction} from two to $n\ge2$
parties, keeping the threshold at two. The left and right keys become $n$
\emph{side keys}, one for each party, and the central key is repeated $n-1$
times.

\subsection{The construction}
Let $n\ge2$ and let $b$ satisfy $2^b>n$. Let $V=\F_{2^b}$ and $B=2^b$, and fix
distinct nonzero elements $\alpha_1,\dots,\alpha_n\in V$, which exist because
$2^b-1\ge n$. For every index $i\in V$ there are $2n-1$ keys: $n-1$
\emph{central} keys $C^{(1)}_i,\dots,C^{(n-1)}_i$ and $n$ \emph{side} keys
$X^{(1)}_i,\dots,X^{(n)}_i$. The key set is
\[
 \K_n=\bigl\{C^{(c)}_i:c\in[n-1],\ i\in V\bigr\}\cup
      \bigl\{X^{(j)}_i:j\in[n],\ i\in V\bigr\},
 \qquad|\K_n|=(2n-1)B .
\]
As in \cref{sec:construction}, the message set is $V$, every share alphabet
is $\{0,1\}\times V$, and we write $x$ for the untagged value $(0,x)$ and
$\bar x$ for the tagged value $(1,x)$. The scheme $\Pi_{b,n}$ gives party
$j\in[n]$ the share
\begin{equation}\label{eq:twon}
 \Shr_j\bigl(m,C^{(c)}_i\bigr)=\alpha_ji+m,\qquad
 \Shr_j\bigl(m,X^{(j')}_i\bigr)=
 \begin{cases}\bar i&\text{if }j=j',\\ \alpha_ji+m&\text{if }j\ne j'.\end{cases}
\end{equation}
Thus every central key gives every party its \emph{true share} $\alpha_ji+m$,
and the side key $X^{(j')}_i$ gives party $j'$ the dummy $\bar i$ and every
other party its true share. The $n-1$ central keys at an index give identical
shares; they differ only as keys. For $n=2$, with $\alpha_1=1$ and
$\alpha_2=\alpha$, this is the scheme $\Pi_b$ of~\eqref{eq:sharing}, with
$C_i=C^{(1)}_i$, $L_i=X^{(2)}_i$, and $R_i=X^{(1)}_i$.

For every $T\subseteq V$ define
\begin{equation}\label{eq:twon-sources}
\begin{gathered}
 \Omega_T=\bigl\{C^{(c)}_i:i\in T,\ c\in[n-1]\bigr\}\cup
 \bigl\{X^{(j)}_i:i\notin T,\ j\in[n]\bigr\},\\
 Z_T=|\Omega_T|=nB-|T|,\qquad P_T=U_{\Omega_T}.
\end{gathered}
\end{equation}
The source class is $\Sour_{b,n}=\{P_T:T\subseteq V\}$. As in
\cref{sec:construction}, at every index a source contains either all central
keys or all side keys, and for $n=2$ these are the sources
of~\eqref{eq:sources}. Since $(n-1)B\le Z_T\le nB$, every $P_T$ is flat with
min-entropy at least $b+\log(n-1)$.

The number of central keys is dictated by privacy. Fix a party $j$. At an
index $i\notin T$, party $j$ receives its true share $\alpha_ji+m$ from the
$n-1$ side keys $X^{(j')}_i$ with $j'\ne j$, and a dummy from $X^{(j)}_i$.
The $n-1$ central keys make it receive its true share equally often at an
index $i\in T$, and therefore every untagged value is equally likely,
whatever $T$ and $m$ are.

\begin{theorem}\label{thm:twon}
Let $n\ge2$ and let $b$ satisfy $2^b>n$. The scheme $\Pi_{b,n}$ is a perfectly
correct two-out-of-$n$ scheme for $b$-bit messages that is perfectly private
for $\Sour_{b,n}$, and:
\begin{enumerate}
\item each share has $b+1$ bits, and the key set $\K_n$ has $(2n-1)2^b$
elements;
\item every distribution in $\Sour_{b,n}$ is uniform on a set of between
$(n-1)2^b$ and $n2^b$ keys, and hence has min-entropy at least
$b+\log(n-1)$;
\item for every $\ell\ge1$ and every function
$\Ext:\K_n\to\{0,1\}^\ell$,
\[
 \max_{P\in\Sour_{b,n}}\TV(\Ext(K),U_\ell)\ \ge\ \frac1{4n-2},
 \qquad K\sim P;
\]
\item for $\ell=1$, some function of the index $i$ alone attains equality in
item~3.
\end{enumerate}
\end{theorem}

For $n=2$ the scheme is $\Pi_b$ and the bound is the $1/6$ of
\cref{thm:main}; item~4 shows that this bound is exact for $\Sour_b$.

\subsection{Correctness, privacy, and non-extractability}
\begin{lemma}[Correctness]\label{lem:correct-n}
For any two distinct parties $j,j'$ there is a map $\Rec_{\{j,j'\}}$ such that
\[
 \Rec_{\{j,j'\}}\bigl(\Shr_j(m,k),\Shr_{j'}(m,k)\bigr)=m
 \qquad\text{for all }m\in V\text{ and }k\in\K_n.
\]
\end{lemma}
\begin{proof}
No key gives a dummy to two parties, so at least one of the two shares is
untagged, and, as in \cref{lem:correct}, the pattern of tags identifies the
case. If party $j$ holds the dummy $\bar i$ and party $j'$ holds
$y=\alpha_{j'}i+m$, then $\Rec_{\{j,j'\}}$ outputs $y+\alpha_{j'}i$; the case
in which party $j'$ holds the dummy is symmetric. If neither share is tagged,
the shares are $x=\alpha_ji+m$ and $y=\alpha_{j'}i+m$, so
$x+y=(\alpha_j+\alpha_{j'})i$ with $\alpha_j\ne\alpha_{j'}$, and
$\Rec_{\{j,j'\}}$ outputs $x+\alpha_j(x+y)/(\alpha_j+\alpha_{j'})$.
\end{proof}

\begin{lemma}[Privacy]\label{lem:private-n}
For every $T\subseteq V$, every $m\in V$, and every party $j$, if $K\sim P_T$
then $\Shr_j(m,K)$ takes each untagged value $x\in V$ with probability
$(n-1)/Z_T$ and each tagged value $\bar i$ with $i\notin T$ with probability
$1/Z_T$. In particular $\Pi_{b,n}$ is perfectly private for $\Sour_{b,n}$.
\end{lemma}
\begin{proof}
Fix $T$, $m$, and $j$. For each $i\in V$, exactly $n-1$ keys of $\Omega_T$ at
index $i$ give party $j$ the share $\alpha_ji+m$: the central keys
$C^{(1)}_i,\dots,C^{(n-1)}_i$ if $i\in T$, and the side keys $X^{(j')}_i$
with $j'\ne j$ if $i\notin T$. Since $\alpha_j\ne0$, the map
$i\mapsto\alpha_ji+m$ is a permutation of $V$, so each untagged value has
probability exactly $(n-1)/Z_T$. The tagged shares of party $j$ come from the
keys $X^{(j)}_i$ with $i\notin T$, one for each such $i$. This distribution does not
depend on $m$.
\end{proof}
For $n=2$ the distribution is uniform, as in \cref{lem:private}; for $n\ge3$ it is not
uniform, but it still does not depend on the message.

\begin{lemma}[No bit can be extracted]\label{lem:noext-n}
For every $\Ext:\K_n\to\{0,1\}$ there is $T\subseteq V$ such that $\Ext$ has
statistical error at least $1/(4n-2)$ on $P_T$ or on $P_{T^c}$. Conversely,
some $\Ext$ that depends only on the index $i$ has error at most $1/(4n-2)$ on
every source in $\Sour_{b,n}$, with equality for some source.
\end{lemma}
\begin{proof}
The first statement is proved exactly as \cref{lem:noext}, with the $n-1$
central keys in place of $C_i$ and the $n$ side keys in place of $L_i,R_i$.
Let $h(k)=(-1)^{\Ext(k)}$ and $H_T=\sum_{k\in\Omega_T}h(k)$, so that $\Ext$
has error $|H_T|/(2Z_T)$ on $P_T$. For each $i\in V$ the integer
\[
 v_i=\sum_{c\in[n-1]}h\bigl(C^{(c)}_i\bigr)-\sum_{j\in[n]}h\bigl(X^{(j)}_i\bigr)
\]
is a sum of $2n-1$ odd numbers and is therefore odd; in particular
$|v_i|\ge1$. With $T=\{i:v_i>0\}$ we get, as in the proof of
\cref{lem:noext},
\[
 H_T-H_{T^c}=\sum_{i\in V}|v_i|\ge B,\qquad Z_T+Z_{T^c}=(2n-1)B,
\]
and~\eqref{eq:bias} becomes
$\max\{|H_T|/Z_T,|H_{T^c}|/Z_{T^c}\}\ge1/(2n-1)$. So $\Ext$ has error at least
$1/(4n-2)$ on $P_T$ or on $P_{T^c}$.

For the second statement, let $g:V\to\{0,1\}$ be balanced and let
$\Ext(k)=g(i)$ for every key $k$ at index $i$. Fix $T$. Under $P_T$, index $i$
carries $w_i$ keys, where $w_i=n-1$ if $i\in T$ and $w_i=n$ otherwise. With
$A_\beta=\sum_{i:\,g(i)=\beta}w_i$ we have
$\Pr[\Ext(K)=1]=A_1/(A_0+A_1)$, and both $A_0$ and $A_1$ lie between
$(n-1)B/2$ and $nB/2$. Hence $\Pr[\Ext(K)=1]$ lies between $(n-1)/(2n-1)$ and
$n/(2n-1)$, so the error is at most $1/(4n-2)$, with equality for
$T=g^{-1}(0)$.
\end{proof}

\Cref{lem:correct-n,lem:private-n,lem:noext-n} prove \cref{thm:twon}: a set
of more than two parties reconstructs from any two of its members, and the
remark after \cref{def:ext} extends \cref{lem:noext-n} to $\ell$ output bits.

\section{\texorpdfstring{$r$}{r}-out-of-\texorpdfstring{$r$}{r} secret sharing}\label{sec:rr}
We now extend our result to $r$-out-of-$r$ secret sharing. Fix $r\ge2$ and
$b\ge r$.  Let $V=\F_{2^b}$, $B=2^b$, and $N=B^{r-1}$, and fix
$\gamma_1,\dots,\gamma_r\in V$ linearly independent over $\F_2$. Such $\gamma_1, \ldots, \gamma_r$ exist
because $b\ge r$. The index is now a vector $u=(u_1,\dots,u_{r-1})\in V^{r-1}$,
and the \emph{pad} of party $j$ is
\[
 p_j(u)=u_j\quad(j<r),\qquad p_r(u)=u_1+\dots+u_{r-1}.
\]
We use two properties of the pads:
\begin{itemize}
\item[(P1)] $p_1(u)+\dots+p_r(u)=0$ for every $u$;
\item[(P2)] for every set $J$ of $r-1$ parties, the map
$p_J(u)=(p_j(u))_{j\in J}$ is a bijection from $V^{r-1}$ to $V^J$.
\end{itemize}
A key is a pair $(u,S)$ with $u\in V^{r-1}$ and $\emptyset\ne S\subseteq[r]$; let $\K^{(r)}$ be the set of keys. The message set is $V$, and the scheme
$\Pi^{(r)}_b$ shares $m\in V$ as
\begin{equation}\label{eq:rr}
 \Shr_j\bigl(m,(u,S)\bigr)=
 \begin{cases}p_j(u)+\gamma_jm&\text{if }j\in S,\\
 \overline{p_j(u)}&\text{if }j\notin S.\end{cases}
\end{equation}
We say that the key $(u,S)$ \emph{activates} the parties in $S$. Thus an
active party receives its pad plus a multiple of the secret, and an inactive
party receives its pad, tagged, as a dummy. For a share $x$ or $\bar x$ we
call $x\in V$ its \emph{value}.

Let $\mathcal O_r$ be the family of subsets of $[r]$ of odd size and
$\mathcal E_r$ the family of nonempty subsets of even size, so that
$|\mathcal O_r|=2^{r-1}$ and $|\mathcal E_r|=2^{r-1}-1$. For every
$T\subseteq V^{r-1}$ define
\begin{equation}\label{eq:rr-sources}
\begin{gathered}
 \Omega_T=\bigl\{(u,S):u\in T,\ S\in\mathcal O_r\bigr\}\cup
 \bigl\{(u,S):u\notin T,\ S\in\mathcal E_r\bigr\},\\
 Z_T=|\Omega_T|=(2^{r-1}-1)N+|T|,\qquad P_T=U_{\Omega_T}.
\end{gathered}
\end{equation}
The source class is $\Sour^{(r)}_b=\{P_T:T\subseteq V^{r-1}\}$. At every index
a source contains either the $2^{r-1}$ keys with an odd set or the
$2^{r-1}-1$ keys with a nonempty even set, just as in
\cref{sec:construction} it contains either $C_i$ or $L_i,R_i$.

For $r=2$ we have $\mathcal O_2=\{\{1\},\{2\}\}$ and
$\mathcal E_2=\{\{1,2\}\}$, and both pads equal $u$. Take $\gamma_1=1$ and
$\gamma_2=\alpha^{-1}$, which are linearly independent over $\F_2$ because
$\alpha\notin\{0,1\}$. After multiplying the untagged shares of party~2 by
$\alpha$, the keys $(u,\{1\})$, $(u,\{2\})$, and $(u,\{1,2\})$ become $L_u$,
$R_u$, and $C_u$ of~\eqref{eq:sharing}, and $P_T$ becomes the source
$P_{V\setminus T}$ of \cref{sec:construction}. For $r=3$ the seven keys at an
index are the three singletons and $\{1,2,3\}$, used at indices in $T$, and
the three pairs, used at indices outside $T$.

\begin{theorem}[$r$-out-of-$r$]\label{thm:rr}
Let $r\ge2$ and $b\ge r$. The scheme $\Pi^{(r)}_b$ is a perfectly correct
$r$-out-of-$r$ scheme for $b$-bit messages that is perfectly private for
$\Sour^{(r)}_b$, and:
\begin{enumerate}
\item each share has $b+1$ bits, and the key set $\K^{(r)}$ has
$(2^r-1)2^{(r-1)b}$ elements;
 
\item every distribution in $\Sour^{(r)}_b$ is uniform on a set of between
$(2^{r-1}-1)N$ and $2^{r-1}N$ keys, and hence has min-entropy at least
$(r-1)b+\log(2^{r-1}-1)$;
\item for every $\ell\ge1$ and every function
$\Ext:\K^{(r)}\to\{0,1\}^\ell$,
\[
 \max_{P\in\Sour^{(r)}_b}\TV(\Ext(K),U_\ell)\ \ge\ \frac1{2(2^r-1)},
 \qquad K\sim P;
\]
\item for $\ell=1$, some function of the index $u$ alone attains equality in
item~3.
\end{enumerate}
\end{theorem}
For $r=2$ the bound is $1/6$, as in \cref{thm:main}, and for $r=3$ it is
$1/14$.

\begin{lemma}[Correctness]\label{lem:correct-rr}
For every $m\in V$ and every key $(u,S)$, the values $z_1,\dots,z_r$ of
the $r$ shares satisfy
\[
 z_1+\dots+z_r=\Bigl(\sum_{j\in S}\gamma_j\Bigr)m,
\]
and the coefficient of $m$ is nonzero. Hence the $r$ shares determine $m$.
\end{lemma}
\begin{proof}
By~\eqref{eq:rr}, $z_j=p_j(u)+\gamma_jm$ for $j\in S$ and $z_j=p_j(u)$ for
$j\notin S$, so the identity follows from (P1). The coefficient
$\sum_{j\in S}\gamma_j$ is nonzero because $S$ is nonempty and the $\gamma_j$
are linearly independent over $\F_2$. The tags reveal $S$, so $\Rec$ computes
this coefficient and divides by it.
\end{proof}

\begin{lemma}[Privacy]\label{lem:private-rr}
Let $J$ be a set of $r-1$ parties. Record the shares of $J$ as the pair
$(R,z)$, where $R=S\cap J$ is the tag pattern and $z\in V^J$ is the vector of
values. For every $T\subseteq V^{r-1}$ and every $m\in V$, if $K\sim P_T$
then $(R,z)$ is uniform on the set
\begin{equation}\label{eq:rr-view}
 \bigl\{(R,z):\emptyset\ne R\subseteq J,\ z\in V^J\bigr\}\ \cup\
 \bigl\{(\emptyset,z):z\in p_J(T)\bigr\},
\end{equation}
which has size $Z_T$ and does not depend on $m$. In particular
$\Pi^{(r)}_b$ is perfectly private for $\Sour^{(r)}_b$.
\end{lemma}
\begin{proof}
Fix $T$ and $m$, and let $j_0$ be the party outside $J$. Fix a nonempty
$R\subseteq J$. A key $(u,S)$ shows $J$ the tag pattern $R$, that is,
$S\cap J=R$, exactly when $S=R$ or $S=R\cup\{j_0\}$.
  These two sets are nonempty and have
different parities, so one of them lies in $\mathcal O_r$ and the other in
$\mathcal E_r$. Hence $\Omega_T$ contains exactly one key at index $u$ that
shows $J$ the pattern $R$, whether or not $u\in T$. That key gives $J$ the
values $p_J(u)+(\ind[j\in R]\,\gamma_jm)_{j\in J}$, which by (P2) run over
$V^J$ exactly once as $u$ runs over $V^{r-1}$. For the empty pattern, the key
must activate $\{j_0\}$, which is odd; so $\Omega_T$ contains such a key at
index $u$ exactly when $u\in T$, and it gives $J$ the values $p_J(u)$,
which do not involve $m$. Altogether, every element of~\eqref{eq:rr-view} is
produced by exactly one key of $\Omega_T$. A coalition of fewer than $r-1$
parties sees a projection of the shares of a coalition of $r-1$ parties, and
so its view does not depend on $m$ either.
\end{proof}

\begin{lemma}[No bit can be extracted]\label{lem:noext-rr}
For every $\Ext:\K^{(r)}\to\{0,1\}$ there is $T\subseteq V^{r-1}$ such that $\Ext$ has statistical error at least
$1/(2(2^r-1))$ on $P_T$ or on $P_{T^c}$. Conversely, some $\Ext$ that depends
only on the index $u$ has error at most $1/(2(2^r-1))$ on every source in
$\Sour^{(r)}_b$, with equality for some source.
\end{lemma}
\begin{proof}
This is the proof of \cref{lem:noext-n}, with the $2^{r-1}$ keys $(u,S)$,
$S\in\mathcal O_r$, in place of the central keys at an index and the
$2^{r-1}-1$ keys $(u,S)$, $S\in\mathcal E_r$, in place of the side keys. Now
\[
 v_u=\sum_{S\in\mathcal O_r}h(u,S)-\sum_{S\in\mathcal E_r}h(u,S)
\]
is a sum of $2^r-1$ odd numbers, hence odd, and $Z_T+Z_{T^c}=(2^r-1)N$, which
gives the first statement. For the second, under every $P_T$ each index
carries $2^{r-1}-1$ or $2^{r-1}$ keys, and the computation in the proof of
\cref{lem:noext-n}, with $n-1$ and $n$ replaced by these two numbers, shows
that a balanced function $g$ of $u$ has error at most $1/(2(2^r-1))$, with
equality for $T=g^{-1}(1)$.
\end{proof}

Combining \Cref{lem:correct-rr,lem:private-rr,lem:noext-rr}, we prove \cref{thm:rr}; the
remark after \cref{def:ext} extends \cref{lem:noext-rr} to $\ell$ output
bits.

\paragraph{Relation to the reductions of~\cite{ACOR22}.}
A $(t,n,t',n')$-distribution design hands each of $n$ new parties the shares
of a set of old parties of a $(t',n')$-scheme, so that any $t$ new parties
together hold the shares of at least $t'$ old parties, and any $t-1$ new
parties hold the shares of fewer than $t'$. It converts every
$(t',n')$-scheme into a $(t,n)$-scheme with the same keys, and hence private
for the same source class. For $t=n=2$ such a design exists whenever
$2\le t'\le n'$: give the first new party the shares of old parties
$1,\dots,t'-1$ and the second the share of old party $t'$. Consequently a
separation for any $(t',n')$ implies one for $(2,2)$. Designs do not go in
the other direction, so neither \cref{thm:twon} for $n\ge3$ nor
\cref{thm:rr} for $r\ge3$ can be derived from \cref{thm:main} in this way. A
$(2,3,2,2)$-design would give three new parties subsets of $\{1,2\}$ of size
at most one whose pairwise unions all equal $\{1,2\}$, and no such subsets
exist; and in an $(r,r,2,2)$-design with $r\ge3$, some $r-1$ of the $r$
subsets would already cover $\{1,2\}$.

\section*{Acknowledgments}

This section is primarily an acknowledgment of defeat. Divesh Aggarwal (the person whose name appears as an author on this paper, henceforth DA) pursued
this problem for more than ten years as a curiosity question, and discussed
it with many people, including Yevgeniy Dodis, Maciej Obremski, Mingyuan
Wang, Eldon Chung, Tomasz Kazana, Arka Rai Choudhuri, and Siyao Guo.

Throughout this period the conjecture was that every source that supports
two-out-of-two secret sharing is extractable, with the lingering thought that
an extractor, if it existed, could not be very complicated, so that the
failure to find one hinted at a counterexample. All attempts to prove the
conjecture or to find a counterexample failed.

In~\cite{ACOR22}, DA and coauthors considered a relaxation in which
the scheme is additionally leakage-resilient or non-malleable, and a
straightforward argument showed that every source supporting such a scheme is
extractable.

Between June and August 2026, DA made several attempts to find an
extractor with ChatGPT~5.5~Pro and ChatGPT~5.6~Pro/Ultra. In most of these
attempts the model was fairly convinced that a deterministic extractor
exists, and it proposed constructions that worked in every case except a
special one.

In September 2026, the same question was posed to ChatGPT~6.0, which, after
several attempts at a proof, stated that the conjecture is false and found a
simple counterexample. Subsequent prompts generalized and cleaned up the
counterexample, with Claude Fable~5.1 joining the effort. The paper was
written with the assistance of Claude Fable~5.1; DA has verified
everything to the best of his ability and takes full responsibility for the text.

\bibliographystyle{alpha}
\bibliography{secret_sharing}

\appendix
\section{Comparison with a statement of Yao and Li}\label{app:comparison}
This appendix concerns only the implication for ordinary secret sharing
in~\cite{YL25}; statement numbers and hypotheses follow the published
journal version, against which this appendix has been checked. We record the
comparison because \cref{thm:main} is incompatible with that implication,
and we indicate the step at which the argument departs from the encryption
case.

Their Definition~4 requires privacy for unauthorized subsets. Definition~5
sets $\psi(m,k)=\Shr(m,k)$, the entire tuple of shares, and calls the scheme
perfectly binding when $\psi(m,k)$ is injective in $m$ for each fixed $k$.
Perfect reconstruction implies this condition. They define
$D_k=\psi(U_b,k)$ and $\mathcal Z=\{D_k:k\in\{0,1\}^d\}$, denoted $\mathcal D$ there. Proposition~2
states that an $\ell$-bit extractor $f$ for $\mathcal Z$ with error
$\varepsilon$ yields an extractor for the original class of key sources with
error $\delta+\varepsilon$, namely
\begin{equation}\label{eq:invalid-reduction}
 \widetilde{\Ext}(k)=f(\psi(0^b,k)).
\end{equation}
The proof substitutes $\psi$ for encryption in~\cite[Lemma~1]{BD07}. That
lemma needs the full output on a fixed message to be close to the full
output on a uniform message. For sharing, perfect reconstruction gives
instead
\begin{equation}\label{eq:joint-distance}
 \TV\bigl(\Shr(0^b,K),\Shr(U_b,K)\bigr)=1-2^{-b}
\end{equation}
for every key distribution, because the distributions of the full tuple under
different messages have disjoint supports. Concretely, for
$1\le\ell\le b$ let $f$ reconstruct the message from the tuple and output its
first $\ell$ bits, with arbitrary values outside the range of $\psi$. Then
$f(D_k)=U_\ell$ for every $k$, while~\eqref{eq:invalid-reduction} is the
constant $0^\ell$. The error is in the statement, not only in the proof:
the scheme of \cref{thm:main} is, in their terminology,
$(\Sour_b,\delta)$-secure for every $\delta>0$, and with this $f$ the family
$\mathcal Z$ is $(1,0)$-extractable, so Proposition~2 would make $\Sour_b$
$(1,\delta)$-extractable for every $\delta>0$, which \cref{thm:main}
excludes for $\delta<1/6$. Theorems~5 and~8 of~\cite{YL25} combine
Proposition~2 with extraction statements for $\mathcal Z$ and fail for
threshold sharing in the same way. 

\end{document}